\documentclass[letterpaper, 10 pt, conference]{ieeeconf}  % Comment this line out
\IEEEoverridecommandlockouts                              % This command is only
\usepackage{graphics} % for pdf, bitmapped graphics files
\usepackage{epsfig} % for postscript graphics files
\usepackage{algorithm,algorithmic}

\usepackage{balance}
\usepackage{amsmath,amsfonts,amssymb,amscd}
\usepackage{capt-of}
\usepackage{bbm}
\usepackage{color}
\usepackage{cite}
\usepackage{verbatim}

\newtheorem{remark}{\bfseries Remark}

\newtheorem{example}{\bfseries Example}

\newtheorem{prop}{\bfseries Proposition}

\newenvironment{varalgorithm}[1]
  {\algorithm\renewcommand{\thealgorithm}{#1}}
  {\endalgorithm}

\begin{document}

\title{\LARGE \bf
Distributed Average Consensus under Quantized Communication \\ via Event-Triggered Mass Splitting}
%\author{ \parbox{3 in}{\centering Huibert Kwakernaak*
%         \thanks{*Use the $\backslash$thanks command to put 
%information here}\\
%         Faculty of Electrical Engineering, Mathematics and 
%Computer Science\\
%         University of Twente\\
%         7500 AE Enschede, The Netherlands\\
%         {\tt\small h.kwakernaak@autsubmit.com}}
%         \hspace*{ 0.5 in}
%         \parbox{3 in}{ \centering Pradeep Misra**
%         \thanks{**The footnote marks may be inserted manually}\\
%        Department of Electrical Engineering \\
%         Wright State University\\
%         Dayton, OH 45435, USA\\
%         {\tt\small pmisra@cs.wright.edu}}
%}

\author{Apostolos~I.~Rikos % <-this % stops a space
% 	\thanks{The authors are with the Department of Electrical and Computer Engineering at the University of Cyprus, Nicosia, Cyprus. E-mails:{\tt~\{arikos01,chadjic\}@ucy.ac.cy}.}
and Christoforos~N.~Hadjicostis % <-this % stops a space 
	\thanks{The authors are with the Department of Electrical and Computer Engineering at the University of Cyprus, Nicosia, Cyprus. E-mails:{\tt~\{arikos01,chadjic\}@ucy.ac.cy}.}
%	\thanks{Christoforos~N.~Hadjicostis is with the Department of Electrical and Computer Engineering at the University of Cyprus, Nicosia, Cyprus. E-mail:{\tt~chadjic@ucy.ac.cy}.}
%\thanks{This work was supported in part by the Cyprus Research Promotion Foundation (CRPF) Framework Programme for Research, Technological Development and Innovation 2009-2010 (CRPF's FP 2009--2010), co-funded by the Republic of Cyprus and the European Regional Development Fund, under \textit{$\Delta I A K P A T I K E \Sigma / K Y - P O Y / 0713 / 21$}. Any opinions, findings, and conclusions or recommendations expressed in this publication are those of the authors and do not necessarily reflect the views of CRPF.}
}
\maketitle
\thispagestyle{empty}
\pagestyle{empty}

%%%%%%%%%%%%%%%%%%%%%%%%%%%%%%%%%%%%%%%%%%%%%%%%%%%%%%%%%%%%%%%%%%%%%%%%%%%%%%%%
\begin{abstract}
We study the distributed average consensus problem in multi-agent systems with directed communication links that are subject to quantized information flow. 
The goal of distributed average consensus is for the nodes, each associated with some initial value, to obtain the average (or some value close to the average) of these initial values. 
In this paper, we present and analyze a distributed averaging algorithm which operates exclusively with quantized values (specifically, the information stored, processed and exchanged between neighboring agents is subject to deterministic uniform quantization) and rely on event-driven updates (e.g., to reduce energy consumption, communication bandwidth, network congestion, and/or processor usage). 
We characterize the properties of the proposed distributed averaging protocol, illustrate its operation with an example, and show that its execution, on any time-invariant and strongly connected digraph, will allow all agents to reach, in finite time, a common consensus value that is equal to the quantized average. We conclude with comparisons against existing quantized average consensus algorithms that illustrate the performance and potential advantages of the proposed algorithm. 
\end{abstract}

\begin{keywords}

Quantized average consensus, event-triggered, distributed algorithms, quantization, digraphs, multi-agent systems. 

\end{keywords}

% ===============================================
%
%
% INTRODUCTION
%
%
% ===============================================
\section{INTRODUCTION}\label{intro}

In recent years, there has been a growing interest for control and coordination of networks consisting of multiple agents, like groups of sensors \cite{2005:XiaoBoydLall} or mobile autonomous agents \cite{2004:Murray}. 
A problem of particular interest in distributed control is the \textit{consensus} problem where the objective is to develop distributed algorithms that can be used by a group of agents in order to reach agreement to a common decision. 
The agents start with different initial values/information and are allowed to communicate locally via inter-agent information exchange under some constraints on connectivity. 
Consensus processes play an important role in many problems, such as leader election \cite{1996:Lynch},  motion coordination of multi-vehicle systems \cite{2005:Olshevsky_Tsitsiklis, 2004:Murray}, and clock synchronization \cite{2007:Gamba}.

%The consensus problem also arises in a number of applications including coordination of UAVs (e.g., aligning the agents’ directions of motion), information processing in sensor networks, and distributed optimization (e.g., agreeing on the estimates of some unknown parameters). 
%The averaging problem is a special case in which the goal is to compute the exact average of the initial values of the agents.

One special case of the consensus problem is distributed averaging, where each agent (initially endowed with a numerical value) can send/receive information to/from other agents in its neighborhood and update its value iteratively, so that eventually, all agents compute the average of the initial values. 
Average consensus is an important problem and has been studied extensively in settings where each agent processes and transmits real-valued states with infinite precision \cite{2018:BOOK_Hadj, 2005:Olshevsky_Tsitsiklis, 2008:Sundaram_Hadjicostis, 2013:Themis_Hadj_Johansson, 2004:XiaoBoyd, 2010:Dimakis_Rabbat, 2011:Morse_Yu, 1984:Tsitsiklis}.

Most existing algorithms, only guarantee asymptotic convergence to the consensus value and cannot be directly applied to real-world control and coordination applications.
Furthermore, in practice, due to constraints on the bandwidth of communication links and the capacity of physical memories, both communication and computation need to be performed assuming finite precision.
For these reasons, researchers have also studied the case when network links can only allow messages of limited length to be transmitted between agents, effectively extending techniques for average consensus towards the direction of quantized consensus.
Various distributed strategies have been proposed, allowing the agents in a network to reach quantized consensus \cite{2007:Aysal_Rabbat, 2012:Lavaei_Murray, 2007:Basar, 2008:Carli_Zampieri, 2016:Chamie_Basar, 2011:Cai_Ishii}. 
Apart from \cite{2016:Chamie_Basar} (which converges in a deterministic manner under a directed communication topology but requires the availability of a set of weights that form a doubly stochastic matrix), these existing strategies use {\em randomized} approaches to address the quantized average consensus problem (implying that all agents reach quantized average consensus with probability one). 
Furthermore, in many types of communication networks it is desirable to update values infrequently to avoid consuming valuable network resources. 
Thus, there has also been an increasing interest for novel event-triggered algorithms for distributed quantized average consensus (and, more generally, distributed control), in order to achieve more efficient usage of network resources \cite{2013:Dimarogonas_Johansson, 2016:nowzari_cortes, 2012:Liu_Chen}.

%The other direction adopts randomized
%time-varying networks with real-valued states, a model that
%potentially captures a variety of random behaviors exhibited
%in realistic networks [14], [18]–[23]; see also [24] for related
%problems in search engines. 

%In such case, one deals with an autonomous discrete-time linear system with a transition matrix $P$, also referred to as weight matrix, that is defined by the coefficients (weights) used in the linear updates. 
%It is well-known that if the weight matrix $P$ is primitive doubly stochastic, then the nodes will asymptotically converge to the average of their initial values (see, e.g., \cite{2004:XiaoBoyd} and references therein). 
%Choosing, in a distributed manner, weights that form a primitive doubly stochastic matrix (and
%can thus be used for asymptotic average consensus) is rather trivial in fixed interconnection topologies that are undirected; however, this is significantly more complex in the case of directed interconnection topologies (digraphs) \cite{2013:Christoforos, 2009:Cortes, 2010:Cortes, 2012:CortesJournal}. 
%An alternative approach, which avoids obtaining a doubly stochastic matrix altogether, is the so-called ratio-consensus \cite{2010:christoforos}, or push-sum algorithm \cite{2003:Kempe, 2012:TsianosLawlorRabbat}. 
%Ratio consensus relies on simultaneously running two iterations (each with different initial conditions) and has each node calculate the exact average by taking the ratio of its two iteration values.

In this paper, we present a novel distributed average consensus algorithm that combines both of the features mentioned above. 
More specifically, the proposed algorithm assumes that the processing, storing, and exchange of information between neighboring agents is ``event-driven'' and  subject to uniform quantization. 
Following \cite{2007:Basar, 2011:Cai_Ishii} we assume that the states are integer-valued (which comprises a class of quantization effects). 
We note that most work dealing with quantization has concentrated on the scenario where the agents have real-valued states but can transmit only quantized values through limited rate channels (see, e.g., \cite{2008:Carli_Zampieri, 2016:Chamie_Basar}). 
By contrast, our assumption is also suited to the case where the states are stored in digital memories of finite capacity (as in \cite{2009:Nedic, 2007:Basar, 2011:Cai_Ishii}) and the control actuation of each node is event-based, which enables more efficient use of available resources. 
The main contribution of this paper is to propose an algorithm that allows all agents to reach quantized consensus in finite time and appears to outperform the current state-of-the-art distributed algorithms for average consensus under quantized communication on directed communication topologies.

%The remainder of this paper is organized as follows. 
%In Section~\ref{notation}, we introduce the notation used throughout the paper, while in Section~\ref{probForm} we formulate the quantized average consensus problem.
%In Section~\ref{algorithm}, we present a probabilistic distributed algorithm, which allows the agents to reach consensus to the quantized average of the initial values, in finite time, with probability one; we also demonstrate its performance with an illustrative example. 
%In Section~\ref{CONValgorithm} we analyze the operation and establish the termination of the proposed algorithm. 
%In Section~\ref{results}, we present simulation results and comparisons against the current state-of-the-art. 
%We conclude in Section~\ref{future} with a brief summary and remarks about future work.

% ===============================================
%
%
% NOTATION
%
%
% ===============================================
\section{PRELIMINARIES}\label{notation}

The sets of real, rational, integer and natural numbers are denoted by $ \mathbb{R}, \mathbb{Q}, \mathbb{Z}$ and $\mathbb{N}$, respectively. 
The symbol $\mathbb{Z}_+$ denotes the set of nonnegative integers and the symbol $\mathbb{N}_0$ denotes the positive natural numbers.

Consider a network of $n$ ($n \geq 2$) agents communicating only with their immediate neighbors. 
The communication topology can be captured by a directed graph (digraph), called \textit{communication digraph}. 
A digraph is defined as $\mathcal{G}_d = (\mathcal{V}, \mathcal{E})$, where $\mathcal{V} =  \{v_1, v_2, \dots, v_n\}$ is the set of nodes (representing the agents of the multi-agent system) and $\mathcal{E} \subseteq \mathcal{V} \times \mathcal{V} - \{ (v_j, v_j) \ | \ v_j \in \mathcal{V} \}$ is the set of edges (self-edges excluded). 
A directed edge from node $v_i$ to node $v_j$ is denoted by $m_{ji} \triangleq (v_j, v_i) \in \mathcal{E}$, and captures the fact that node $v_j$ can receive information from node $v_i$ (but not the other way around). 
We assume that the given digraph $\mathcal{G}_d = (\mathcal{V}, \mathcal{E})$ is \textit{static}\footnote{In this paper we assume that the given digraph is static, however the operation of the proposed protocol can also be extended for jointly connected dynamic topologies (i.e., digraphs whose structure changes over time but their union graphs over consecutive large time intervals remain strongly connected).} (i.e., does not change over time) and \textit{strongly connected} (i.e., for each pair of nodes $v_j, v_i \in \mathcal{V}$, $v_j \neq v_i$, there exists a directed \textit{path} from $v_i$ to $v_j$). 
The subset of nodes that can directly transmit information to node $v_j$ is called the set of in-neighbors of $v_j$ and is represented by $\mathcal{N}_j^- = \{ v_i \in \mathcal{V} \; | \; (v_j,v_i)\in \mathcal{E}\}$, while the subset of nodes that can directly receive information from node $v_j$ is called the set of out-neighbors of $v_j$ and is represented by $\mathcal{N}_j^+ = \{ v_l \in \mathcal{V} \; | \; (v_l,v_j)\in \mathcal{E}\}$. 
The cardinality of $\mathcal{N}_j^-$ is called the \textit{in-degree} of $v_j$ and is denoted by $\mathcal{D}_j^-$ (i.e., $\mathcal{D}_j^- = | \mathcal{N}_j^- |$), while the cardinality of $\mathcal{N}_j^+$ is called the \textit{out-degree} of $v_j$ and is denoted by $\mathcal{D}_j^+$ (i.e., $\mathcal{D}_j^+ = | \mathcal{N}_j^+ |$).

In the proposed distributed protocol we assume that each node is aware of its out-neighbors and can directly (or indirectly\footnote{Indirect transmission could involve broadcasting a message to all out-neighbors while including in the message header the ID of the out-neighbor it is intended for.}) transmit messages to each out-neighbor (but, cannot necessarily receive messages from them). 
Furthermore, each node $v_j$ assigns a nonzero probability $b_{lj}$ to each of its outgoing edges $m_{lj}$ (including a virtual self-edge), where $v_l \in \mathcal{N}^+_j \cup \{ v_j \}$. 
This probability assignment for all nodes can be captured by a column stochastic matrix $\mathcal{B} = [b_{lj}]$. 
A very simple choice would be to set these probabilities to be equal, i.e.,
\begin{align*}
b_{lj} = \left\{ \begin{array}{ll}
         \frac{1}{1 + \mathcal{D}_j^+}, & \mbox{if $v_{l} \in \mathcal{N}_j^+ \cup \{v_j\}$,}\\
         0, & \mbox{otherwise.}\end{array} \right. 
\end{align*}
Each nonzero entry $b_{lj}$ of matrix $\mathcal{B}$ represents the probability of node $v_j$ transmitting towards out-neighbor $v_l \in \mathcal{N}^+_j$ through the edge $m_{lj}$, or performing no transmission\footnote{From the definition of $\mathcal{B} = [b_{lj}]$ we have that $b_{jj} = \frac{1}{1 + \mathcal{D}_j^+}$, $\forall v_j \in \mathcal{V}$. This represents the probability that node $v_j$ will not perform a transmission to any of its out-neighbors $v_l \in \mathcal{N}^+_j$ (i.e., it will transmit to itself).}.

% ===============================================
%
%
% PROBLEM
%
%
% ===============================================
\section{PROBLEM FORMULATION}\label{probForm}

Consider a strongly connected digraph $\mathcal{G}_d = (\mathcal{V}, \mathcal{E})$, where each node $v_j \in \mathcal{V}$ has an initial (i.e., for $k=0$) quantized value $y_j[0]$ (for simplicity, we take $y_j[0] \in \mathbb{Z}$). 
In this paper, we develop a distributed algorithm that allows nodes (while processing and transmitting \textit{quantized} information via available communication links between nodes) to eventually obtain, after a finite number of steps, a quantized value $q^s$ which is equal to the ceiling $q^s = \lceil q \rceil$ or the floor $q^s = \lfloor q \rfloor$ of the actual average $q$ of the initial values, where
\begin{equation}\label{real_av}
q = \frac{\sum_{l=1}^{n}{y_l[0]}}{n} .
\end{equation}
Note that $q$ will in general be a real (rational) number.

\begin{remark}
Following \cite{2007:Basar, 2011:Cai_Ishii} we assume that the state variables maintained at each node are integer valued. 
This abstraction subsumes a class of quantization effects (e.g., uniform quantization).
\end{remark}

The quantized average $q^s$ is defined as the ceiling $q^s = \lceil q \rceil$ or the floor $q^s = \lfloor q \rfloor$ of the true average $q$ of the initial values. 
Let $S \triangleq \mathbf{1}^{\rm T} y[0]$, where $\mathbf{1} = [1 \ ... \ 1]^{\rm T}$ is the vector of all ones, and let $y[0] = [y_1[0] \ ... \ y_n[0]]^{\rm T}$ be the vector of the quantized initial values. 
We can write $S$ uniquely as $S = nL + R$ where $L$ and $R$ are both integers and $0 \leq R < n$. 
Thus, we have that either $L$ or $L+1$ may be viewed as an integer approximation of the average of the initial values $S/n$ (which may not be integer in general).

The algorithm we develop are iterative. 
With respect to quantization of information flow, we have that at time step $k \in \mathbb{Z}_+$ (where $\mathbb{Z}_+$ is the set of nonnegative integers), each node $v_j \in \mathcal{V}$ maintains five variables, namely the state variables $y^s_j, z^s_j, q_j^s$, where $y^s_j \in \mathbb{Z}$, $z^s_j \in \mathbb{N}_0$ and $q_j^s \in \mathbb{Z}$ (where $q_j^s = \lfloor \frac{y_j^s}{z_j^s} \rfloor$ or $q_j^s = \lceil \frac{y_j^s}{z_j^s} \rceil$), and the mass variables $y_j, z_j$ where $y_j \in \mathbb{Z}$ and $z_j \in \mathbb{N}$. 
The aggregate states are denoted by $y^s[k] = [y^s_1[k] \ ... \ y^s_n[k]]^{\rm T} \in \mathbb{Z}^n$, $z^s[k] = [z^s_1[k] \ ... \ z^s_n[k]]^{\rm T} \in \mathbb{N}_0^n$, $q^s[k] = [q^s_1[k] \ ... \ q^s_n[k]]^{\rm T} \in \mathbb{Z}^n$ and $y[k] = [y_1[k] \ ... \ y_n[k]]^{\rm T} \in \mathbb{Z}^n$, $z[k] = [z_1[k] \ ... \ z_n[k]]^{\rm T} \in \mathbb{N}^n$ respectively. 

Following the execution of the proposed distributed algorithm, we argue that $\exists \ k_0$ so that for every $k \geq k_0$ we have
\begin{equation}\label{alpha_q}
q^s_j[k] = \lfloor q \rfloor \ \ \ \text{or} \ \ \ q^s_j[k] = \lceil q \rceil
\end{equation}
for every $v_j \in \mathcal{V}$ where $q$, from (\ref{real_av}), is the actual average of the initial values.

%Consider a strongly connected digraph $\mathcal{G}_d = (\mathcal{V}, \mathcal{E})$, where each node $v_j \in \mathcal{V}$ has initial value (i.e., for $k=0$) mass variables $y_j[0], z_j[0]$ where $y_j[0] \in \mathbb{Z}$ and $z_j[0] = 1$, and initial state variables $y^s_j[0], z^s_j[0], q_j^s[0]$, where $y^s_j[0] = y_j[0]$, $z^s_j[0] = z_j[0]$, and $q_j^s[0] = \frac{y_j^s[0]}{z_j^s[0]}$. 
%In this paper we develop a distributed algorithm that allows nodes to process and transmit \textit{quantized} information via available (pairwise) communication links between nodes, so that the nodes eventually obtain, after a finite number of steps, a quantized fraction $q^s$ which is equal to the average $q$ of the initial values, where
%\begin{equation}
%q = \frac{\sum_{j=1}^{n}{y_j[0]}}{n} .
%\end{equation}

% ===============================================
%
%
% ALGORITHM PROBABILISTIC EVENT TRIGGER
%
%
% ===============================================
\section{QUANTIZED AVERAGING ALGORITHM WITH MASS SPLITTING}\label{algorithm}

In this section we propose a probabilistic distributed information exchange process in which the nodes transmit and receive quantized messages so that they reach quantized average consensus on their initial values after a finite number of steps.

\noindent
The operation of the proposed distributed algorithm is summarized below.

\noindent
\textbf{Initialization:}
Each node $v_j$ selects a set of probabilities $\{ b_{lj} \ | \ v_{l} \in \mathcal{N}_j^+ \cup \{v_j\} \}$ such that $0 < b_{lj} < 1$ and $\sum_{v_{l} \in \mathcal{N}_j^+ \cup \{v_j\}} b_{lj} = 1$ (see Section~\ref{notation}). 
Each value $b_{lj}$, represents the probability for node $v_j$ to transmit towards out-neighbor $v_l \in \mathcal{N}^+_j$ (or transmits towards itself), at any given time step (independently between time steps and between different nodes).  
Each node has some initial value $y_j[0] \in \mathbb{Z}$, and also sets its mass variable, for time step $k=0$, as $z_j[0] = 1$. 
%Then, according to the nonzero probability $b_{lj}$, node $v_j$ either transmits $z_j[0]$ and $y_j[0]$ towards an out-neighbor $v_l \in \mathcal{N}_j^+$ or performs no transmission. 
%If it performed a transmission towards an out-neighbor, it sets $y_j[0] = 0$ and $z_j[0] = 0$. 

\noindent
The iteration involves the following steps:

\noindent
\textbf{Step 1. Event Trigger Condition:}
Node $v_j$ checks the following condition
$$
z_j[k] > 0 . 
$$
If the above condition holds, node $v_j$ sets $z^s_j[k] = z_j[k]$, $y^s_j[k] = y_j[k]$ and 
$$
q^s_j[k] = \Bigl \lfloor \frac{y^s_j[k]}{z^s_j[k]} \Bigr \rfloor . 
$$
Then, it splits $y_j[k]$ in $z_j[k]$ equal pieces (or with maximum difference between them equal to $1$) , which we denote by $y^{(t)}_j[k]$, $t = 1, 2, ..., z_j[k]$. 
Specifically, node $v_j$ sets $y^{(t)}_j[k] = \lfloor y_j[k] / z_j[k] \rfloor$ (or $y^{(t)}_j[k] = \lceil y_j[k] / z_j[k] \rceil$) and $z^{(t)}_j[k] = 1$ (with $t$ taking integer values from $1$ to $z_j[k]$) so that $\sum_{t = 1}^{z_j[k]} y^{(t)}_j[k] = y_j[k]$ and $\sum_{t = 1}^{z_j[k]} z^{(t)}_j[k] = z_j[k]$. 
Furthermore, an additional requirement in this splitting is that the difference between $y^{(t)}_j[k]$ for different values of $t$ is equal to $0$ or $1$ (i.e., $\vert y^{(t)}_j[k] - y^{(t')}_j[k] \vert \leq 1$, for $t, t' \in \{ 1, 2, ..., z_j[k] \}$). 

\noindent
\textbf{Step 2. Transmitting:} If the ``Event Trigger Conditions'' above hold, for each set of values $y^{(t)}_j[k]$, $z^{(t)}_j[k]$, node $v_j$ uses the nonzero probabilities $b_{lj}$ (assigned by node $v_j$ during the initialization step), in order to transmit $y^{(t)}_j[k]$, $z^{(t)}_j[k]$ towards out-neighbor $v_l \in \mathcal{N}_j^+$ or towards itself. 
Each time, it chooses an out-neighbor or itself randomly, independently from other values of $t$, other nodes, or previous time steps.

\noindent
\textbf{Step 3. Receiving:} Each node $v_j$ receives messages $y^{(t)}_i[k]$ and $z^{(t)}_i[k]$ from its in-neighbors $v_i \in \mathcal{N}_j^-$, and it sums them along with any of its own stored messages (i.e., the sets of values it transmitted to itself) as
$$
y_j[k+1] = \sum_{v_i \in \mathcal{N}_j^- \cup \{v_j\}} \sum_{t = 1}^{z_i[k]} w^{(t)}_{ji}[k] \ y^{(t)}_i[k] ,
$$
and 
$$
z_j[k+1] = \sum_{v_i \in \mathcal{N}_j^- \cup \{v_j\}} \sum_{t = 1}^{z_i[k]} w^{(t)}_{ji}[k] \ z^{(t)}_i[k] ,
$$
where $w^{(t)}_{ji}[k] = 0$ if split message $t$ was not sent to node $v_j$ from in-neighbor $v_i \in \mathcal{N}_j^-$; otherwise $w^{(t)}_{ji}[k] = 1$. 
Then, $k$ is set to $k+1$ and the iteration repeats (it goes back to Step~1).

\begin{remark}
Although not discussed in this paper, asynchronous operation is not an issue for the proposed probabilistic distributed protocol. 
Moreover, communication disturbances such as (time-varying and inhomogeneous) time delays, that might affect transmissions between different agents in the network, may also be addressed.
\end{remark}

The probabilistic quantized mass transfer process is detailed as Algorithm~\ref{algorithm_prob} below (for the case when $b_{lj} = 1/(1+\mathcal{D}_j^+)$ for $v_l \in \mathcal{N}_j^+ \cup \{ v_j \}$ and $b_{lj}=0$ otherwise). 
We next provide an example to illustrate the operation of the proposed distributed protocol.

\noindent
\vspace{-0.5cm}    
\begin{varalgorithm}{1}
\caption{Quantized Average Consensus via Mass Splitting}
\textbf{Input} 
\\ 1) A strongly connected digraph $\mathcal{G}_d = (\mathcal{V}, \mathcal{E})$ with $n=|\mathcal{V}|$ nodes and $m=|\mathcal{E}|$ edges. 
\\ 2) For every $v_j$ we have $y_j[0] \in \mathbb{Z}$. 
\\
\textbf{Initialization} 
\\ Every node $v_j \in \mathcal{V}$: 
\\ 1) Assigns a nonzero probability $b_{lj}$ to each of its outgoing edges $m_{lj}$, where $v_l \in \mathcal{N}^+_j \cup \{v_j\}$, as follows  
\begin{align*}
b_{lj} = \left\{ \begin{array}{ll}
         \frac{1}{1 + \mathcal{D}_j^+}, & \mbox{if $l = j$ or $v_{l} \in \mathcal{N}_j^+$,}\\
         0, & \mbox{if $l \neq j$ and $v_{l} \notin \mathcal{N}_j^+$.}\end{array} \right. 
\end{align*}
\\ 2) Sets $z_j[0] = 1$. 
\\
\textbf{Iteration}
\\ For $k=0,1,2,\dots$, each node $v_j \in \mathcal{V}$ does the following:
\\ 1) \underline{Event Trigger Condition:} If the following condition holds,
\begin{equation}\nonumber
z_j[k] > 0,
\end{equation}
it performs the following two steps: 
\\ a) It sets $z^s_j[k] = z_j[k]$, $y^s_j[k] = y_j[k]$, which means that 
$$
q^s_j[k] = \Bigl \lfloor \frac{y^s_j[k]}{z^s_j[k]} \Bigr \rfloor \ . 
$$
Then, for $t \in \{1, 2, ..., z_j[k]\}$, it sets $y^{(t)}_j[k] = \lfloor y_j[k] / z_j[k] \rfloor$ and $z^{(t)}_j[k] = 1$. 
If $r \equiv y_j[k] - z_j[k] \lfloor y_j[k]/z_j[k] \rfloor$ is nonzero, then node $v_j$ increases by one the value of $y_j^{(t)}[k]$, $t = 1, 2, ..., r$, so that $\sum_{t = 1}^{z_j[k]} y^{(t)}_j[k] = y_j[k]$ and $\sum_{t = 1}^{z_j[k]} z^{(t)}_j[k] = z_j[k]$. 
Furthermore, for $t, t' \in \{ 1, 2, ..., z_j[k] \}$ it also holds that $\vert y^{(t)}_j[k] - y^{(t')}_j[k] \vert \leq 1$. 
\\ b) For each $t \in \{1, 2, ..., z_j[k]\}$, it transmits the set of values $y^{(t)}_j[k]$, $z^{(t)}_j[k]$ towards a randomly chosen out-neighbour $v_l \in \mathcal{N}_j^+$ or towards itself.
\\ 2) It receives $y^{(t)}_i[k]$ and $z^{(t)}_i[k]$ from its in-neighbours $v_i \in \mathcal{N}_j^-$ and from itself and sets 
$$
y_j[k+1] = \sum_{v_i \in \mathcal{N}_j^- \cup \{v_j\}} \sum_{t = 1}^{z_i[k]} w^{(t)}_{ji}[k] \ y^{(t)}_i[k] ,
$$
and 
$$
z_j[k+1] = \sum_{v_i \in \mathcal{N}_j^- \cup \{v_j\}} \sum_{t = 1}^{z_i[k]} w^{(t)}_{ji}[k] \ z^{(t)}_i[k] ,
$$
where $w^{(t)}_{ji}[k] = 1$ if node $v_j$ receives value $y^{(t)}_i[k]$ and $z^{(t)}_i[k]$ from node $v_i$ at iteration $k$ (otherwise $w^{(t)}_{ji}[k] = 0$). 
\\ 3) It repeats (increases $k$ to $k + 1$ and goes back to Step~1).
\label{algorithm_prob}
\end{varalgorithm}

% \subsection{Example of Algorithm Operation}

\begin{example}\label{Ex1}

Consider the strongly connected digraph $\mathcal{G}_d = (\mathcal{V}, \mathcal{E})$ shown in Fig.~\ref{prob_example} (borrowed from \cite{2018:RikosHadj_CDC}), with $\mathcal{V} = \{ v_1, v_2, v_3, v_4 \}$ and $\mathcal{E} = \{ m_{21}, m_{31}, m_{42}, m_{13}, m_{23}, m_{34} \}$, where each node has initial quantized values $y_1[0] = 5$, $y_2[0] = 3$, $y_3[0] = 7$, and $y_4[0] = 2$ respectively. 
The actual average $q$ of the initial values of the nodes, is equal to $q = 4.25$ which means that the quantized value $q^s$ is equal to $q^s = 4$ or $q^s = 5$ (i.e., the ceiling or the floor of the average $q$).

\begin{figure}[h]
\begin{center}
\includegraphics[width=0.25\columnwidth]{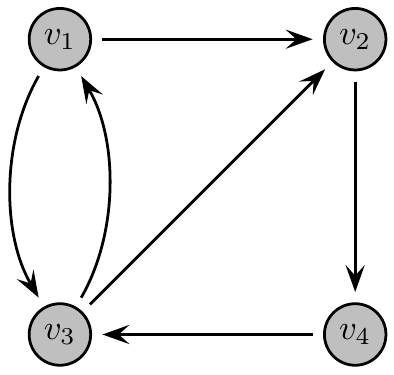}
\caption{Example of digraph for probabilistic quantized averaging.}
\label{prob_example}
\end{center}
\end{figure}

Each node $v_j \in \mathcal{V}$ follows the Initialization steps ($1-2$) in Algorithm~\ref{algorithm_prob}, assigning to each of its outgoing edges $v_l \in \mathcal{N}_j^+ \cup \{v_j\}$ a nonzero probability value $b_{lj}$ equal to $b_{lj} = \frac{1}{1 + \mathcal{D}_j^+}$. 
The assigned values can be seen in the following matrix
\[
\mathcal{B}
=
\begin{bmatrix}
    \frac{1}{3} & 0 & \frac{1}{3} &  0\\ \vspace{-0.35cm} \\
    \frac{1}{3} & \frac{1}{2} & \frac{1}{3} &  0\\ \vspace{-0.35cm} \\
    \frac{1}{3} & 0 & \frac{1}{3} &  \frac{1}{2}\\ \vspace{-0.35cm} \\
    0 & \frac{1}{2} & 0 &  \frac{1}{2}\\ \vspace{-0.35cm} \\
\end{bmatrix}
.
\] 
Furthermore, each node $v_j \in \mathcal{V}$ sets $z_j[0] = 1$.

For the execution of the proposed algorithm, at time step $k=0$, each node $v_j$ calculates its state variables $y^s_j[0]$, $z^s_j[0]$ and $q^s_j[0]$. 
The mass and state variables for $k=0$ are shown in Table~\ref{tableProb}. 
Then, every node $v_j$ calculates the values it will transmit.
Specifically, nodes $v_1, v_2, v_3, v_4$ set $y^{(1)}_1[0] = 5$, $y^{(1)}_2[0] = 3$, $y^{(1)}_3[0] = 7$, $y^{(1)}_4[0] = 2$ and $z^{(1)}_1[0] = 1$, $z^{(1)}_2[0] = 1$, $z^{(1)}_3[0] = 1$, $z^{(1)}_4[0] = 1$, respectively. 
Then, suppose that nodes $v_1$, $v_3$ and $v_4$ transmit to nodes $v_2$, $v_1$ and $v_3$, respectively, whereas node $v_2$, performs no transmission (i.e., transmits to itself).

For the execution of the proposed algorithm, each node $v_j$ receives from its in-neighbors $v_i \in \mathcal{N}_j^- \cup \{v_j\}$ the transmitted mass variables $y_i[0]$ and $z_i[0]$ and then, at time step $k=1$, it calculates its state variables $y^s_j[1]$, $z^s_j[1]$ and $q^s_j[1]$. 
The mass and state variables for $k=1$ are shown in Table~\ref{tableProbk_1}. 
Here we have that nodes $v_1$ and $v_3$ have mass variables $y_1[1] = y_3[0] = 7$, $z_1[1] = z_3[0] = 1$ and $y_3[1] = y_4[0] = 2$, $z_3[1] = z_4[0] = 1$ (and update their state variables), while node $v_2$ has mass variables $y_2[1] = y_1[0] + y_2[0] = 8$, $z_2[1] = z_1[0] + z_2[0] = 2$ (also updating its state variables). 
Then, every node $v_j$ calculates the values it will transmit (notice that node $v_2$ will split its mass variable $y_2[1]$ in two equal pieces since $z_2[1] = 2$). 
Specifically, we have that $y^{(1)}_1[1] = 7$, $y^{(1)}_2[1] = 4$, $y^{(2)}_2[1] = 4$, $y^{(1)}_4[1] = 2$ and $z^{(1)}_1[1] = 1$, $z^{(1)}_2[1] = 1$, $z^{(2)}_2[1] = 1$, $z^{(1)}_4[1] = 1$, respectively. 
Then, suppose that nodes $v_1$ and $v_3$ both transmit to node $v_2$, while node $v_2$, transmits the set of values $y^{(1)}_2[1]$, $z^{(1)}_2[1]$ to itself and the set of values $y^{(2)}_2[1]$, $z^{(2)}_2[1]$ to $v_4$.

\begin{center}
\captionof{table}{Initial Mass and State Variables for Fig.~\ref{prob_example}}
\label{tableProb}
{\small 
\begin{tabular}{|c||c|c|c|c|c|}
\hline
Nodes &\multicolumn{5}{c|}{Mass and State Variables for $k=0$}\\
$v_j$ &$y_j[0]$&$z_j[0]$&$y^s_j[0]$&$z^s_j[0]$&$q^s_j[0]$\\
\cline{2-6}
 &  &  &  &  & \\
$v_1$ & 5 & 1 & 5 & 1 & 5\\
$v_2$ & 3 & 1 & 3 & 1 & 3\\
$v_3$ & 7 & 1 & 7 & 1 & 7\\
$v_4$ & 2 & 1 & 2 & 1 & 2\\
\hline
\end{tabular}
}
\end{center}
\vspace{0.4cm}

\begin{center}
\captionof{table}{Mass and State Variables for Fig.~\ref{prob_example} for $k=1$}
\label{tableProbk_1}
{\small 
\begin{tabular}{|c||c|c|c|c|c|}
\hline
Nodes &\multicolumn{5}{c|}{Mass and State Variables for $k=1$}\\
$v_j$ &$y_j[1]$&$z_j[1]$&$y^s_j[1]$&$z^s_j[1]$&$q^s_j[1]$\\
\cline{2-6}
 &  &  &  &  & \\
$v_1$ & 7 & 1 & 7 & 1 & 7\\
$v_2$ & 8 & 2 & 8 & 2 & 4\\
$v_3$ & 2 & 1 & 2 & 1 & 2\\
$v_4$ & 0 & 0 & 2 & 1 & 2\\
\hline
\end{tabular}
}
\end{center}
\vspace{0.2cm}

Each node $v_j$ receives from its in-neighbors the transmitted mass variables and, at time step $k=2$, it calculates its state variables $y^s_j[2]$, $z^s_j[2]$ and $q^s_j[2]$ (which are shown in Table~\ref{tableProbk_2}). 
Then, every node $v_j$ calculates the values it will transmit as $y^{(1)}_2[2] = 4$, $y^{(2)}_2[2] = 4$, $y^{(3)}_2[2] = 5$, $y^{(1)}_4[1] = 4$ and $z^{(1)}_2[2] = 1$, $z^{(2)}_2[2] = 1$, $z^{(3)}_2[2] = 1$, $z^{(1)}_4[1] = 1$, respectively. 
It is interesting to notice here that all the calculated values $y^{(t)}_j[2]$ are equal to the quantized average of the initial values (i.e., the ceiling or the floor of the real average $q = 4.25$). 
Then, suppose that node $v_4$ transmits to node $v_3$, while node $v_2$, transmits the set of values $y^{(1)}_2[2]$, $z^{(1)}_2[2]$ and $y^{(2)}_2[2]$, $z^{(2)}_2[2]$ to $v_4$ and the set of values $y^{(3)}_2[2]$, $z^{(3)}_2[2]$ to itself.

\begin{center}
\captionof{table}{Mass and State Variables for Fig.~\ref{prob_example} for $k=2$} 
\label{tableProbk_2}
{\small 
\begin{tabular}{|c||c|c|c|c|c|}
\hline
Nodes &\multicolumn{5}{c|}{Mass and State Variables for $k=2$}\\
$v_j$ &$y_j[2]$&$z_j[2]$&$y^s_j[2]$&$z^s_j[2]$&$q^s_j[2]$\\
\cline{2-6}
 &  &  &  &  & \\
$v_1$ & 0 & 0 & 7 & 1 & 7\\
$v_2$ & 13 & 3 & 13 & 3 & 4\\
$v_3$ & 0 & 0 & 2 & 1 & 2\\
$v_4$ & 4 & 1 & 4 & 1 & 4\\
\hline
\end{tabular}
}
\end{center}
\vspace{0.2cm}

Each node $v_j$ receives from its in-neighbors the transmitted mass variables and, at time step $k=3$, it calculates its state variables $y^s_j[3]$, $z^s_j[3]$ and $q^s_j[3]$ (which are shown in Table~\ref{tableProbk_3}). 
Then, every node $v_j$ calculates the values it will transmit as $y^{(1)}_2[3] = 5$, $y^{(1)}_3[3] = 4$, $y^{(1)}_4[3] = 4$, $y^{(2)}_4[3] = 4$ and $z^{(1)}_2[3] = 1$, $z^{(1)}_3[3] = 1$, $z^{(1)}_4[3] = 1$, $z^{(2)}_4[3] = 1$, respectively. 
Then, suppose that nodes $v_2$ and $v_3$ transmit to node $v_1$ and $v_4$, while node $v_4$, transmits the set of values $y^{(1)}_4[3]$, $z^{(1)}_4[3]$ and $y^{(2)}_4[3]$, $z^{(2)}_4[3]$ to node $v_3$.

Next, each node $v_j$ receives from its in-neighbors the transmitted mass variables and, at time step $k=4$, it calculates its state variables $y^s_j[4]$, $z^s_j[4]$ and $q^s_j[4]$ which are shown in Table~\ref{tableProbk_4}.

\begin{center}
\captionof{table}{Mass and State Variables for Fig.~\ref{prob_example} for $k=3$}
\label{tableProbk_3}
{\small 
\begin{tabular}{|c||c|c|c|c|c|}
\hline
Nodes &\multicolumn{5}{c|}{Mass and State Variables for $k=3$}\\
$v_j$ &$y_j[3]$&$z_j[3]$&$y^s_j[3]$&$z^s_j[3]$&$q^s_j[3]$\\
\cline{2-6}
 &  &  &  &  & \\
$v_1$ & 0 & 0 & 7 & 1 & 7\\
$v_2$ & 5 & 1 & 5 & 1 & 5\\
$v_3$ & 4 & 1 & 4 & 1 & 4\\
$v_4$ & 8 & 2 & 8 & 2 & 4\\
\hline
\end{tabular}
}
\end{center}
\vspace{0.4cm}

\begin{center}
\captionof{table}{Mass and State Variables for Fig.~\ref{prob_example} for $k=4$}
\label{tableProbk_4}
{\small 
\begin{tabular}{|c||c|c|c|c|c|}
\hline
Nodes &\multicolumn{5}{c|}{Mass and State Variables for $k=4$}\\
$v_j$ &$y_j[4]$&$z_j[4]$&$y^s_j[4]$&$z^s_j[4]$&$q^s_j[4]$\\
\cline{2-6}
 &  &  &  &  & \\
$v_1$ & 4 & 1 & 4 & 1 & 4\\
$v_2$ & 5 & 1 & 5 & 1 & 5\\
$v_3$ & 8 & 2 & 8 & 2 & 4\\
$v_4$ & 0 & 0 & 8 & 2 & 4\\
\hline
\end{tabular}
}
\end{center}
\vspace{0.2cm}

From Table~\ref{tableProbk_4}, we can see that for $k \geq 4$ it holds that 
$$
q_j^s[k] = \lfloor q \rfloor = 4 , \ \ \ \text{or} \ \ \ q_j^s[k] = \lceil q \rceil = 5 ,
$$
for every $v_j \in \mathcal{V}$, which means that every node $v_j$ obtained, after a finite number of iterations, a quantized value $q_j^s$, which is equal to the ceiling or the floor of the real average $q$ of the initial values of the nodes. 
The state variable $q_j^s[k]$ of every node $v_j \in \mathcal{V}$ can also be seen in Figure~\ref{example_plot}, in which we can see that, after a finite number of time steps $k$, it holds that $q_j^s[k] = 4$ or $q_j^s[k] = 5$. 
\end{example}

\begin{figure} [ht]
\centering
\includegraphics[width=70mm]{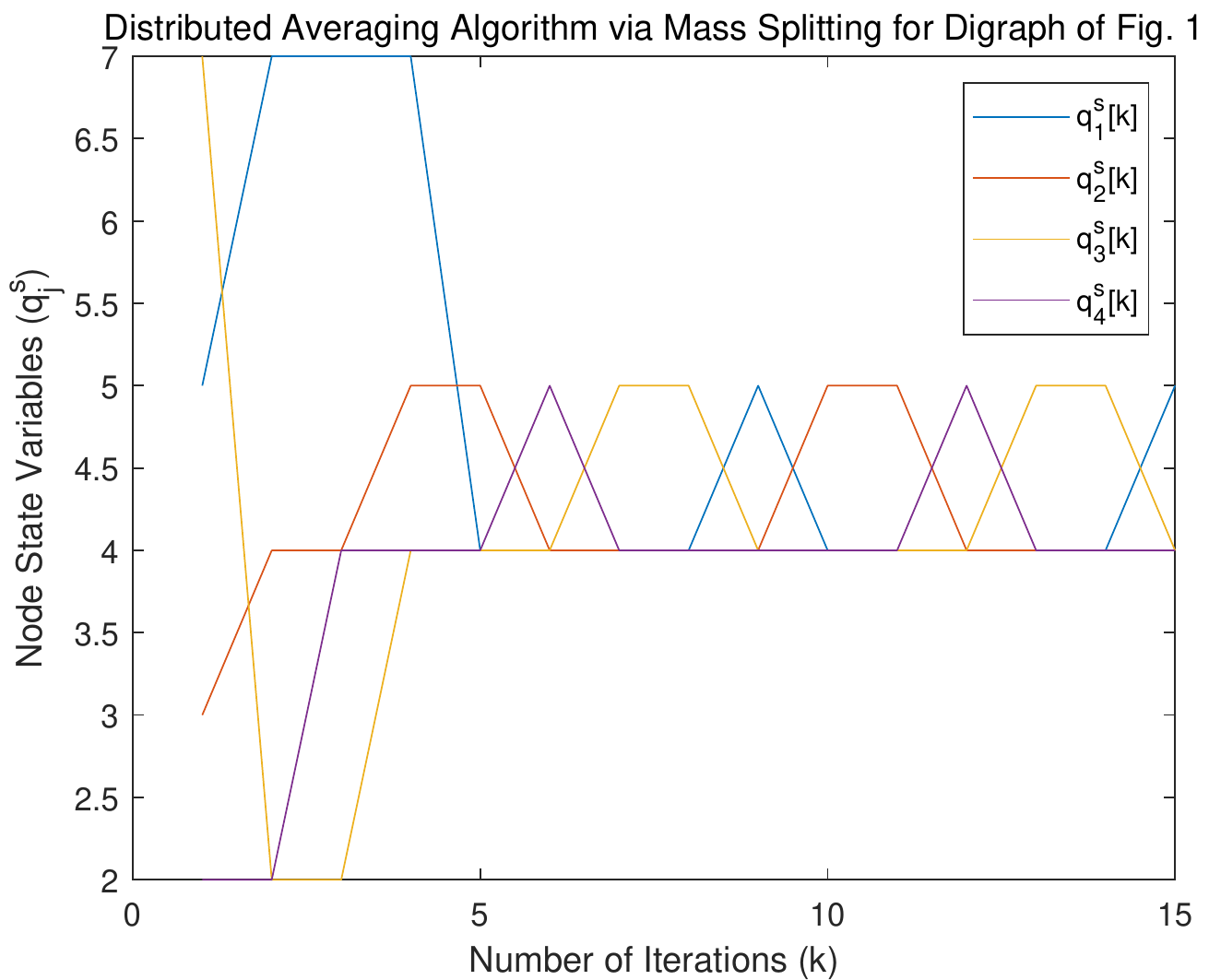}
\caption{Node state variables plotted against the number of iterations for Algorithm~\ref{algorithm_prob} for the digraph shown in Fig.~\ref{prob_example}.}
\label{example_plot}
\end{figure}

\begin{remark}
Notice that the operation of Algorithm~\ref{algorithm_prob} is different from the algorithms presented in \cite{2018:RikosHadj_CDC}. 
Specifically, in \cite{2018:RikosHadj_CDC}, the authors presented two distributed algorithms (a probabilistic and a deterministic algorithm) in which every node $v_j$ ``merged'' (i.e., added) the incoming mass variables (which remained ``merged'' through the algorithm execution), sent by its in-neighbours. 
The authors showed that every node $v_j$ calculated, after a finite number of time steps, a quantized fraction which is equal to the actual average $q$ of the initial values of the nodes (i.e., there was zero quantization error), but due to strict accumulation of the values, the proposed protocol required a significant amount of time steps. 
During the operation of Algorithm~\ref{algorithm_prob}, every node $v_j$, is able to calculate, after a finite number of steps, a quantized value which is equal to the ceiling or the floor of the initial average (i.e., there is a nonzero quantization error defined as the difference between the actual average $q$ and the quantized average $q^s$), but, as we will see in the following sections, its operation outperforms (in terms of convergence speed) the ones presented in \cite{2018:RikosHadj_CDC} along with the state-of-the-art algorithms in the available literature. 
\end{remark}

\section{CONVERGENCE OF MASS SPLITTING ALGORITHM}\label{CONValgorithm}

We are now ready to prove that, during the operation of Algorithm~\ref{algorithm_prob}, each agent $v_j$ reaches, after a finite number of time steps, a consensus value which is equal to the ceiling or the floor of the actual average $q$ of the initial values of the nodes. 
We present the following proposition which is necessary for our subsequent development. 
Due to space limitations, we do not provide the proof for Proposition~\ref{PROP1_prob} below; it will be made available in an extended
version of this paper.

\begin{prop}
\label{PROP1_prob}
Consider a strongly connected digraph $\mathcal{G}_d = (\mathcal{V}, \mathcal{E})$ with $n=|\mathcal{V}|$ nodes and $m=|\mathcal{E}|$ edges.
Suppose that each node assigns a nonzero probability $b_{lj}$ to each of its outgoing edges $m_{lj}$, where $v_l \in \mathcal{N}^+_j \cup \{v_j\}$, as follows  
\begin{align*}
b_{lj} = \left\{ \begin{array}{ll}
         \frac{1}{1 + \mathcal{D}_j^+}, & \mbox{if $l = j$ or $v_{l} \in \mathcal{N}_j^+$,}\\
         0, & \mbox{if $l \neq j$ and $v_{l} \notin \mathcal{N}_j^+$,}\end{array} \right. 
\end{align*}
and, at time step $k=0$, node $v_j$ holds a ``token" while the other nodes $v_l \in \mathcal{V} - \{ v_j \}$ do not. 
Each node $v_j$ transmits the ``token'' (if it has it, otherwise it performs no transmission) according to the nonzero probability $b_{lj}$ it assigned to its outgoing edges $m_{lj}$. 
The probability that the token is at node $v_i$ after $n-1$ time steps satisfies 
\begin{equation}\label{lowerProf}
\text{P}_{\text{Token at node $v_i$ at step $n-1$}} \geq (1+D^+_{max})^{-(n-1)} , 
\end{equation}
where $D^+_{max} = \max_{v_j \in \mathcal{V}} D^+_{j}$. 
\end{prop}

\begin{prop}
\label{PROP2_prob}
Consider a strongly connected digraph $\mathcal{G}_d = (\mathcal{V}, \mathcal{E})$ with $n=|\mathcal{V}|$ nodes and $m=|\mathcal{E}|$ edges and $z_j[0] = 1$ and $y_j[0] \in \mathbb{Z}$ for every node $v_j \in \mathcal{V}$ at time step $k=0$. 
Suppose that each node $v_j \in \mathcal{V}$ follows the Initialization and Iteration steps as described in Algorithm~\ref{algorithm_prob}. 
With probability one, there exists $k_0 \in \mathbb{Z}_+$, so that for every $k \geq k_0$ we have 
$$
q^s_j[k] = \lceil q \rceil \ \ \text{or} \ \ q^s_j[k] = \lfloor q \rfloor ,
$$
for every $v_j \in \mathcal{V}$ (i.e., for $k \geq k_0$ every node $v_j$ has calculated the ceiling or the floor of the actual average $q$ of the initial values). 
\end{prop}

\begin{proof}
During the Initialization Step~$1$ of Algorithm~\ref{algorithm_prob}, each node $v_j \in \mathcal{V}$ assigns a nonzero probability $b_{lj} = \frac{1}{1 + \mathcal{D}_j^+}$ to each of its outgoing edges $m_{lj}$, where $v_l \in \mathcal{N}^+_j \cup \{v_j\}$. 
We can consider the digraph $\mathcal{G}_d = (\mathcal{V}, \mathcal{E})$ with associated transition matrix $\mathcal{B} = [b_{lj}]$ as a Markov chain in which the nodes of the graph are equivalent to the states of the Markov chain and the weight $b_{lj}$ of matrix $\mathcal{B}$ represents the probability of a transition from node $v_j$ towards node $v_l$. 
It is important to notice that during Iteration Steps~$1$ and $2$, each node $v_j$, splits the received messages $y_j[k]$, $z_j[k]$ into $z_j[k]$ equal (or with maximum difference equal to $1$) pieces $y^{(t)}_j[k]$, $z^{(t)}_j[k]$, where $y^{(t)}_j[k] \in \mathbb{Z}$ and $z^{(t)}_j[k] = 1$ for $t = 1, 2, ..., z_j[k]$. 
Then it transmits each set of messages $y^{(t)}_j[k]$, $z^{(t)}_j[k]$ towards a randomly chosen out-neighbour $v_l \in \mathcal{N}_j^+ \cup \{v_j\}$ according to the nonzero probabilities $b_{lj}$ (assigned during the initialization step). 
This means that the operation of the Algorithm~\ref{algorithm_prob} can be interpreted as the ``random walk'' of $n$ ``tokens'' in a Markov chain, where $n=|\mathcal{V}|$, and each `token contains a set of values $y[k]$, $z[k]$, for which $y[k] \in \mathbb{Z}$ and $z[k] = 1$, during each time step $k$. 

During the operation of Algorithm~\ref{algorithm_prob}, from Iteration Step~$1$, we have that if two ``tokens'' meet in the same node (say $v_j$), during time step $k$, then their values $y[k]$ become equal (or with maximum difference equal to $1$). Furthermore, the sum of the $y_j[k]$ values at any given $k$ is equal to the initial sum (i.e., $\sum_{j=1}^n y_j[k] = \sum_{j=1}^n y_j[0]$). 
Thus, we will focus on the scenario in which all $n$ tokens meet at a common node and obtain equal values $y[k]$ (or with maximum difference between them equal to $1$). 

From Proposition~\ref{PROP1_prob}, we have that after $n-1$ time steps, the probability that {\em one} ``token'' is at node $v_i$ is 
$$
\text{P}_{\text{Token at node $v_i$ at step $n-1$}} \geq (1+D^+_{max})^{-(n-1)} .
$$
Considering that, during the operation of Algorithm~\ref{algorithm_prob}, the $n$ ``tokens'' perform {\em independent} random walks we have that the probability that all $n$ tokens meet at node $v_i$ after $n-1$ time steps is
$$
\text{P}_{\text{All tokens at node $v_i$ at step $n-1$}} \geq (1+D^+_{max})^{-n(n-1)} .
$$
Furthermore, since the events ``all tokens meet at node $v_i$ after $n-1$ time steps'' and ``all tokens meet at node $v_j$ after $n-1$ time steps'' are mutually exclusive (i.e., they have a zero intersection) then we have that the probability that all tokens meet at {\em any} node $v_j \in \mathcal{V}$ after $n-1$ time steps is
\begin{eqnarray}\nonumber
\text{P}_{\text{All tok. at node $v_i$ at step $n-1$}} & \geq & \sum_{v_j \in \mathcal{V}} (1+D^+_{max})^{-n(n-1)} \Rightarrow \nonumber \\
\text{P}_{\text{All tok. at node $v_i$ at step $n-1$}} & \geq & n (1+D^+_{max})^{-n(n-1)} .   \nonumber \;
\end{eqnarray}
This means that, for the scenario ``not all tokens meet at any node after $n-1$ time steps'' we have
\begin{equation}\label{ProbNotMeet}
\text{P}_{\text{Not all tok. at any node at step $n-1$}} \leq 1 - n (1+D^+_{max})^{-n(n-1)} .
\end{equation}
Note that $\text{P}_{\text{Not all tok. at any node at step $n-1$}}$ denotes the probability that no node will receive all $n$ tokens after $n-1$ time steps. 

By extending the above analysis we have that after $\tau (n-1)$ time steps (i.e., $\tau$ windows, each one consisting of $n-1$ time steps), we have that the probability that ``not all tokens meet at any node after $\tau$ time steps'' is
\begin{equation}\label{ProbNotMeet_after_t}
\text{P}_{\text{Not all tok. at any node after $\tau$}} \leq [\text{P}_{\text{Not all tok. at any node at step $n-1$}}]^{\tau} .
\end{equation}
Since, from (\ref{ProbNotMeet}), we have that $P_{not \ all} < 1$ this means that, by executing Algorithm~\ref{algorithm_prob} for $\tau$ time windows, from (\ref{ProbNotMeet_after_t}) we have that
\begin{equation}\label{ProbNotMeet_after_t_to_zero}
\lim_{\tau \rightarrow \infty} \text{P}_{\text{Not all tok. at any node after $\tau$}} = 0 .
\end{equation}
As a result, with probability $1$, we have that $\exists k_0' \in \mathbb{Z}$ for which all $n$ ``tokens'' meet at node $v_j$.
This means that all $n$ ``tokens'' will have equal values $y[k_0']$ (or with maximum differences between them equal to $1$).  
Furthermore, from Iteration Step~$1$, we have that each node $v_j$ splits $y_j[k]$ in $z_j[k]$ equal (or with maximum difference between them equal to $1$) pieces $y^{(t)}_j[k]$, $z^{(t)}_j[k]$, where $y^{(t)}_j[k] \in \mathbb{Z}$ and $z^{(t)}_j[k] = 1$ for $t = 1, 2, ..., z_j[k]$ for which it holds that $\sum_{t = 1}^{z_j[k]} y^{(t)}_j[k] = y_j[k]$ and $\sum_{t = 1}^{z_j[k]} z^{(t)}_j[k] = z_j[k]$. 
This means that $\sum_{j=1}^{n} y[k_0'] = \sum_{j=1}^{n} y[0]$ and we have that the $y[k_0']$ values of each ``token'' will become equal to the ceiling or the floor of the actual average $q$ of the initial values (i.e., $y[k_0'] = \lfloor q \rfloor$ or $y[k_0'] = \lceil q \rceil$). 

Continuing the operation of Algorithm~\ref{algorithm_prob}, we have that, for time steps $k > k_0'$, the $n$ ``tokens'' will continue performing random walks in the digraph $\mathcal{G}_d$. 
This means that, since $\mathcal{G}_d$ is strongly connected, we have that $\exists k_0 \in \mathbb{N}$, where $k_0 > k_0'$, for which every node $v_j \in \mathcal{V}$ will receive (at least once) one (or multiple) ``tokens'' during the time interval $(k_0', k_0]$. 
From Iteration Step~$1$, this means that the state variables $q_j[k_0]$ of every node $v_j \in \mathcal{V}$, will be equal to the ceiling or the floor of the actual average $q$ (i.e., $q_j[k_0] = \lceil q \rceil$ or $q_j[k_0] = \lfloor q \rfloor$, for every $v_j \in \mathcal{V}$) which completes the proof of this proposition.  
\end{proof}

\section{SIMULATION RESULTS} \label{results}

In this section, we present simulation results and comparisons. 
Specifically, we present simulation results of the proposed distributed algorithm for the digraph $\mathcal{G}_d = (\mathcal{V}, \mathcal{E})$ (borrowed from \cite{2014:RikosHadj}), shown in Fig.~\ref{simul}, with $\mathcal{V} = \{ v_1, v_2, v_3, v_4, v_5, v_6, v_7 \}$ and $\mathcal{E} = \{ m_{21}, m_{51}, m_{12}, m_{52}, m_{13}, m_{53}, m_{24}, m_{54}, m_{65}, m_{75},$ $m_{36}, m_{47}, m_{67} \}$, where each node has initial quantized values $y_1[0] = 15$, $y_2[0] = 5$, $y_3[0] = 11$, $y_4[0] = 4$, $y_5[0] = 3$, $y_6[0] = 13$, and $y_7[0] = 9$, respectively. 
The real average $q$ of the initial values of the nodes, is equal to $q = \frac{57}{7} = 8.57$ which means that the quantized average $q^s$ is equal to $q^s = 8$ or $q^s = 9$. 

\vspace{-.3cm}

\begin{figure}[h]
\begin{center}
\includegraphics[width=0.40\columnwidth]{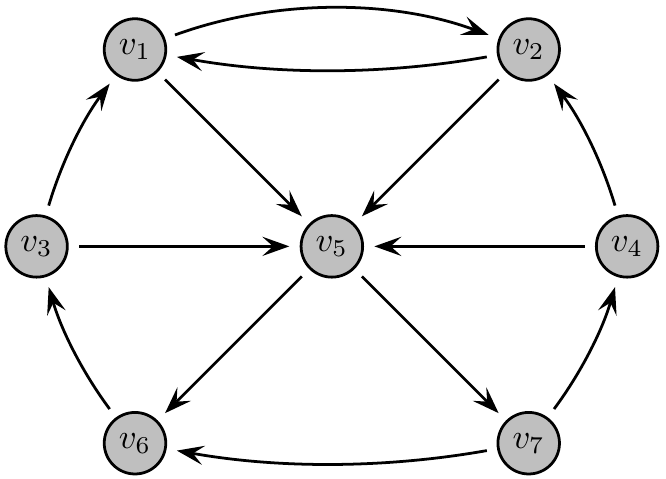}
\caption{Example of digraph for simulation of Algorithm~\ref{algorithm_prob}.}
\label{simul}
\end{center}
\end{figure}

\vspace{-.3cm}

In Figure~\ref{simul_plot} we plot the state variable $q_j^s[k]$ of every node $v_j \in \mathcal{V}$ as a function of the number of iterations $k$ for the digraph shown in Fig.~\ref{simul}. 
The plot demonstrates that Algorithm~\ref{algorithm_prob} is able to achieve a common quantized consensus value to the average of the initial states after a finite number of iterations.

\begin{figure} [ht]
\centering
\includegraphics[width=75mm]{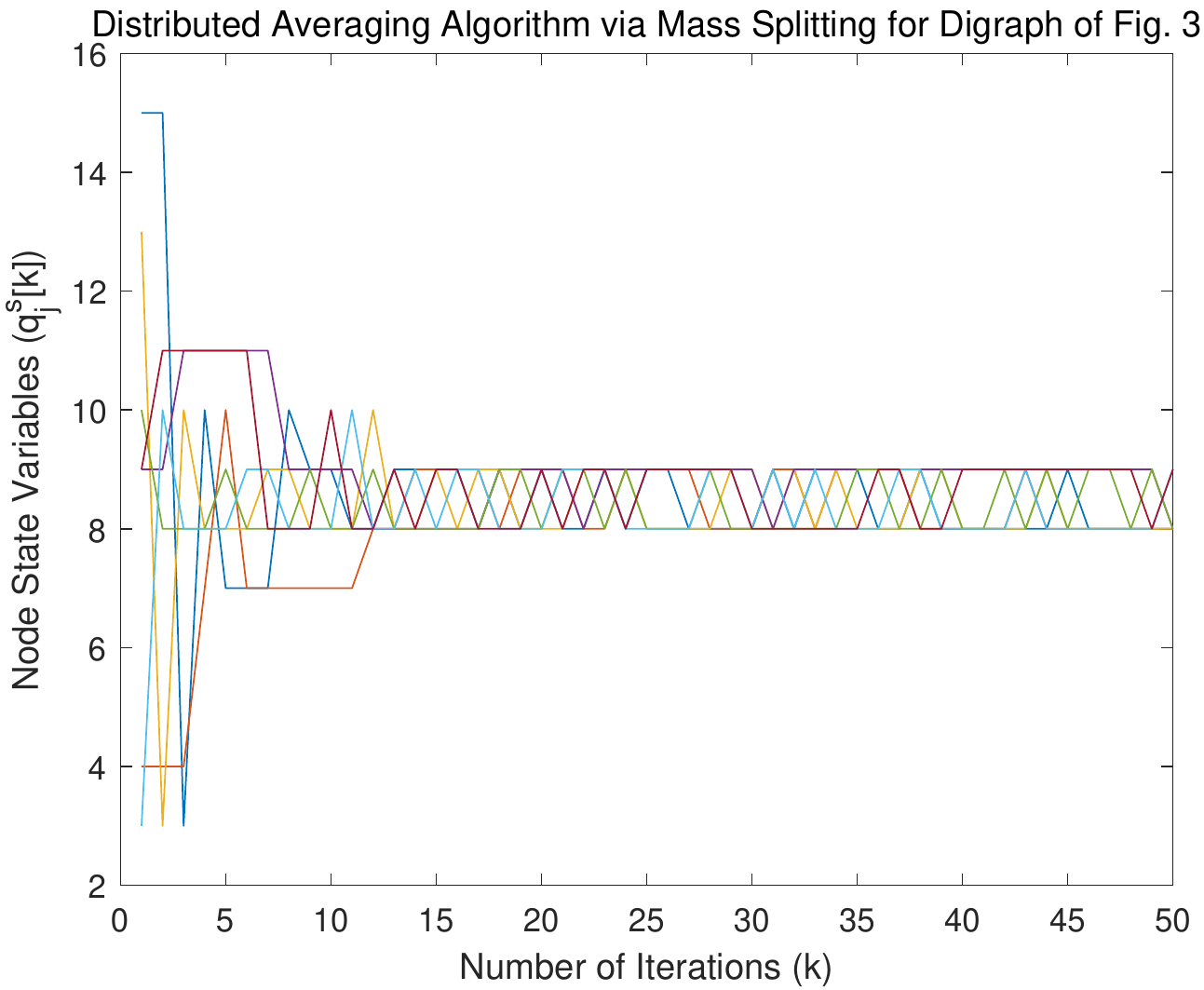}
\caption{Node state variables plotted against the number of iterations for Algorithm~\ref{algorithm_prob} for the digraph shown in Fig.~\ref{simul}.}
\label{simul_plot}
\end{figure}

Now, we compare its performance against four other algorithms: (a) the quantized gossip algorithm presented in \cite{2007:Basar} in which, at each time step $k$, one edge\footnote{Note here that the algorithm presented in \cite{2007:Basar} requires the underlying graph to be undirected. For this reason, in Figure~\ref{comp20}, we consider, for the algorithm in \cite{2007:Basar}, the underlying graph to be undirected (i.e., if $(v_j, v_i) \in \mathcal{E}$ then also $(v_i, v_j) \in \mathcal{E}$) while, for the algorithms in \cite{2011:Cai_Ishii, 2016:Chamie_Basar, 2018:RikosHadj_CDC} we consider the underlying graph to be directed.} is selected at random, independently from earlier instants and the values of the nodes that the selected edge is incident on are updated, (b) the quantized asymmetric averaging algorithm presented in \cite{2011:Cai_Ishii} in which, at each time step $k$, one edge, say edge $(v_l, v_j)$, is selected at random and, node $v_j$ sends its state information and surplus and node $v_l$ performs updates 
over its own state and surplus values, (c) the distributed averaging algorithm with quantized communication presented in \cite{2016:Chamie_Basar} in which, at each time step $k$, each agent $v_j$ broadcasts a quantized version of its own state value towards its out-neighbors, (d) the distributed averaging algorithm with quantized communication presented in \cite{2018:RikosHadj_CDC} in which, at each time step $k$, each agent sends its mass variables towards a randomly chosen out-neighbor in the form of a quantized fraction.

Figure~\ref{comp20} presents a study of the case of $1000$ digraphs of $20$ nodes each, in which the average of the nodes initial values is equal to $q = \dfrac{651}{20} = 32.55$. 
The results shown are averaged over $1000$ graphs. 
The top of Figure~\ref{comp20} suggests that the operation of Algorithm~\ref{algorithm_prob} outperforms the quantized distributed algorithms in the available literature \cite{2007:Basar, 2011:Cai_Ishii, 2016:Chamie_Basar, 2018:RikosHadj_CDC}.

\begin{figure*} [ht]
\centering
\includegraphics[width=55mm]{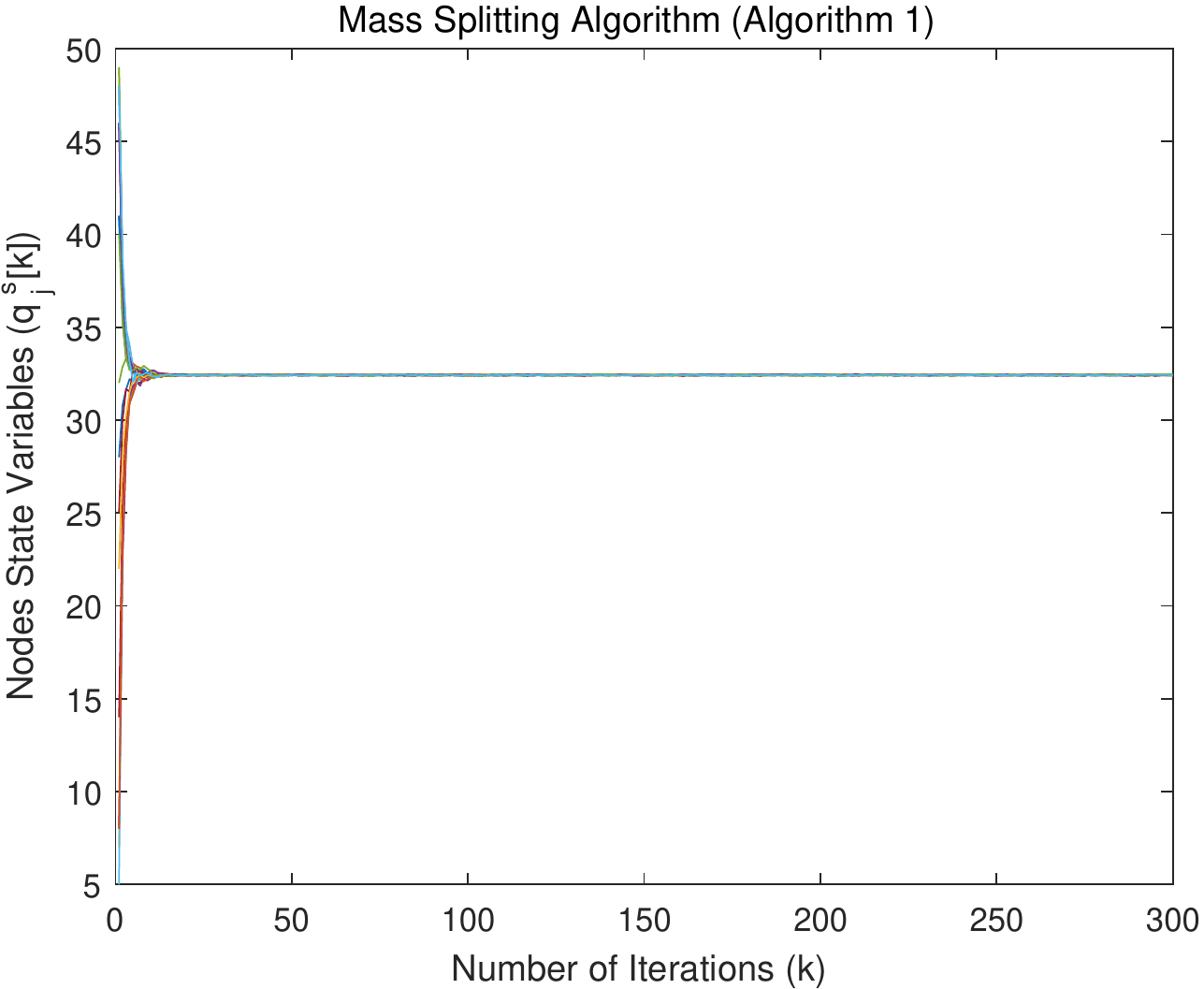}~~\hspace*{0.5cm}
\includegraphics[width=55mm]{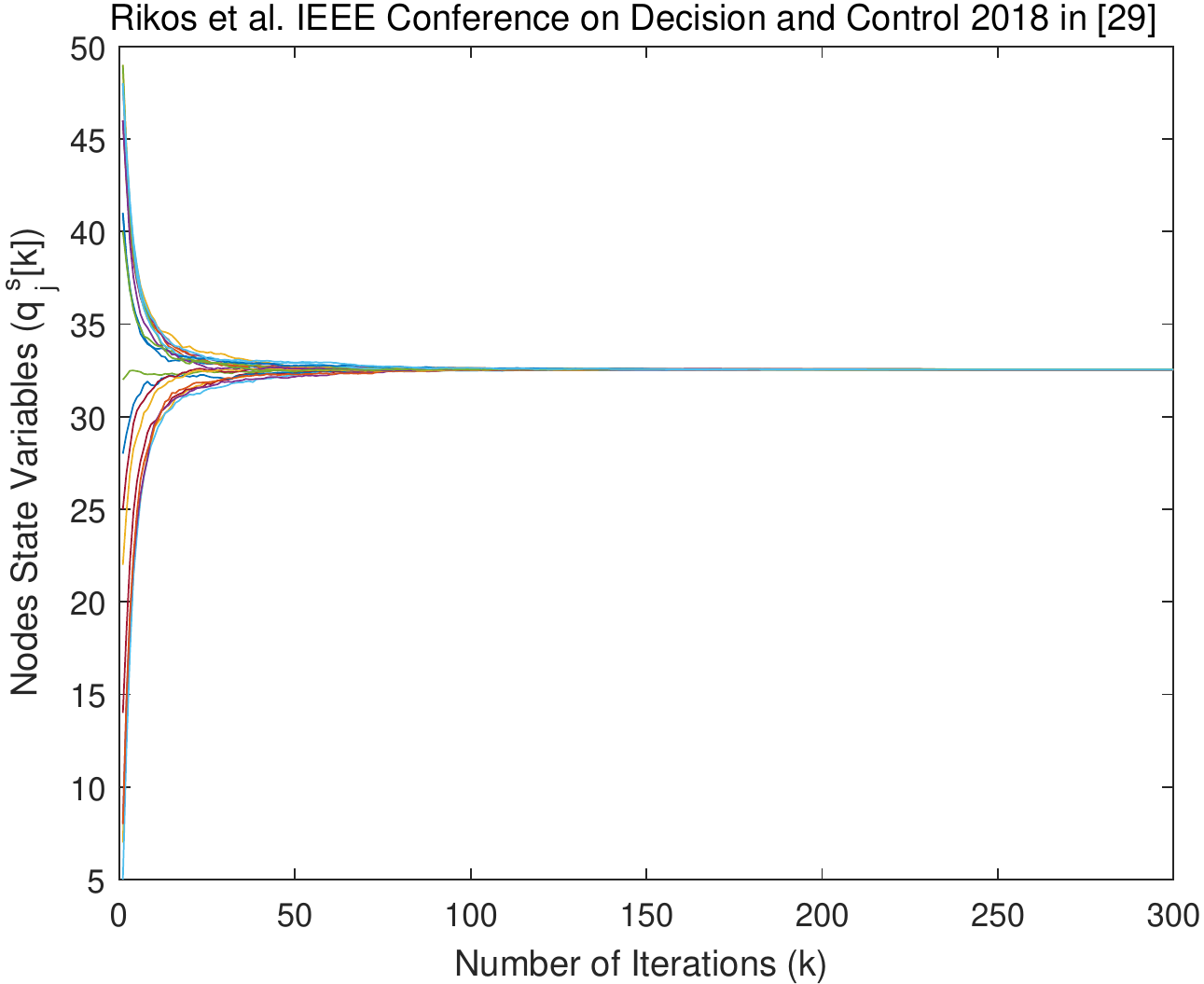} \\ \vspace{.2cm}
\includegraphics[width=55mm]{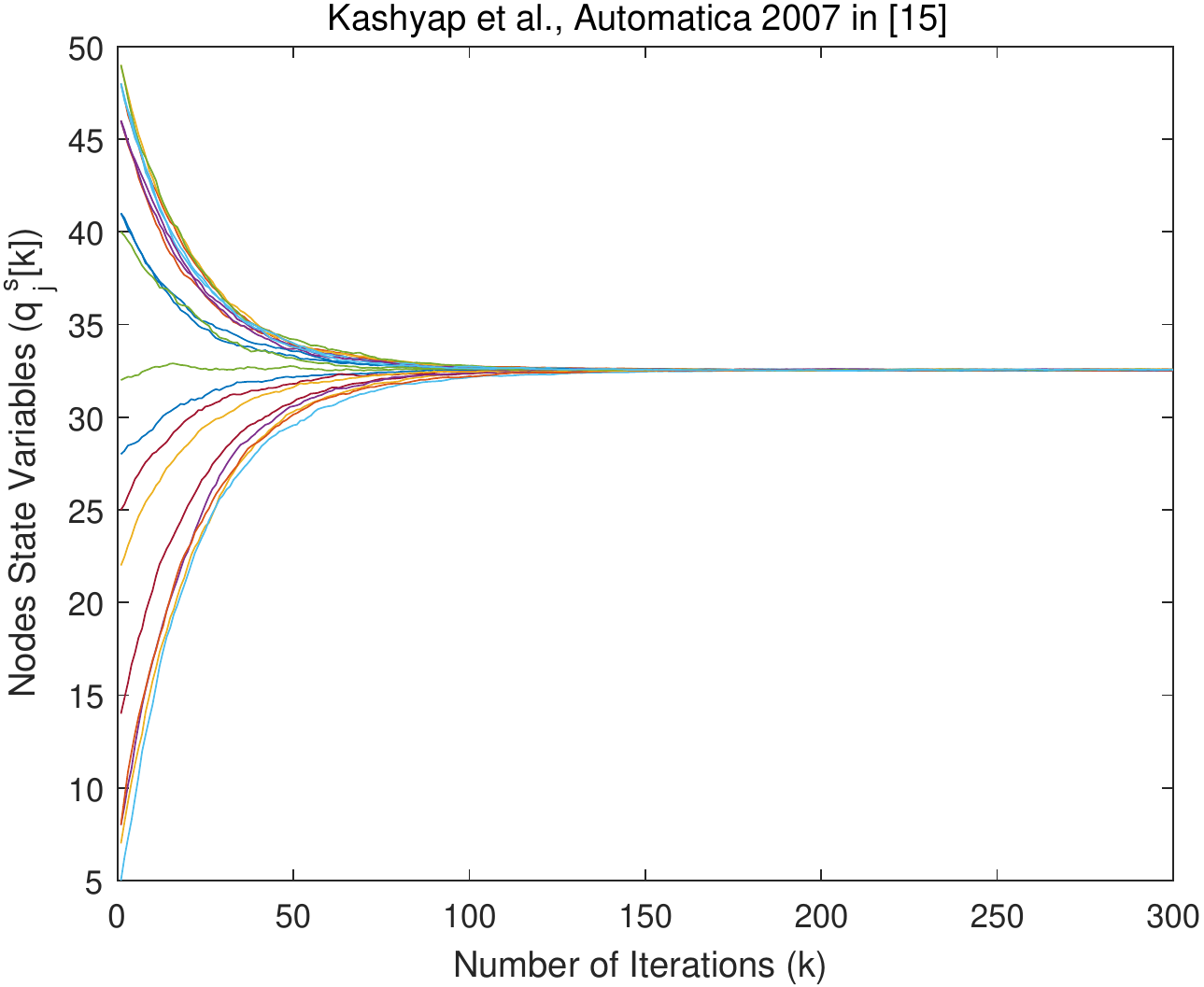}~~\hspace*{0.5cm}
\includegraphics[width=55mm]{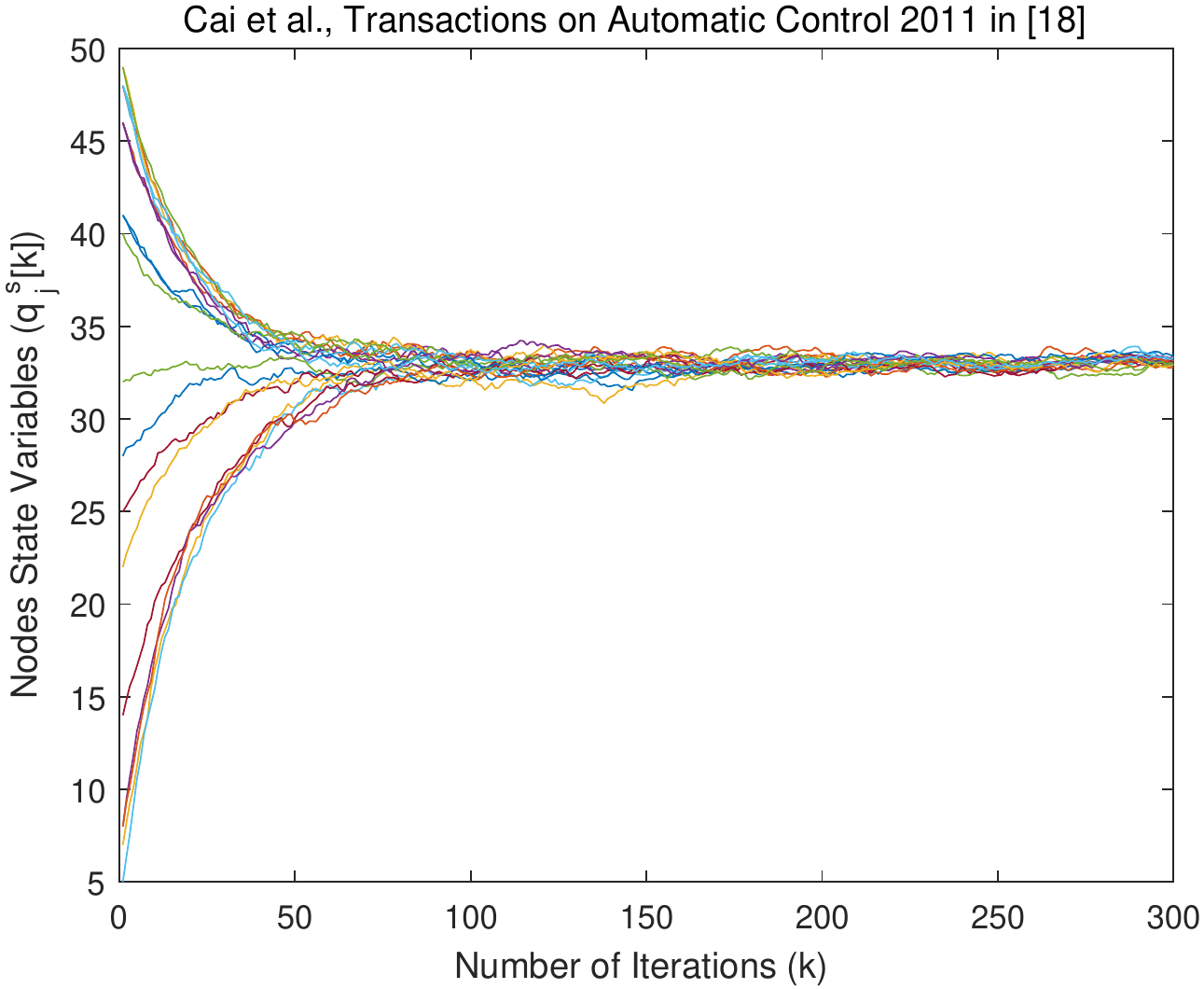}~~\hspace*{0.5cm}
\includegraphics[width=55mm]{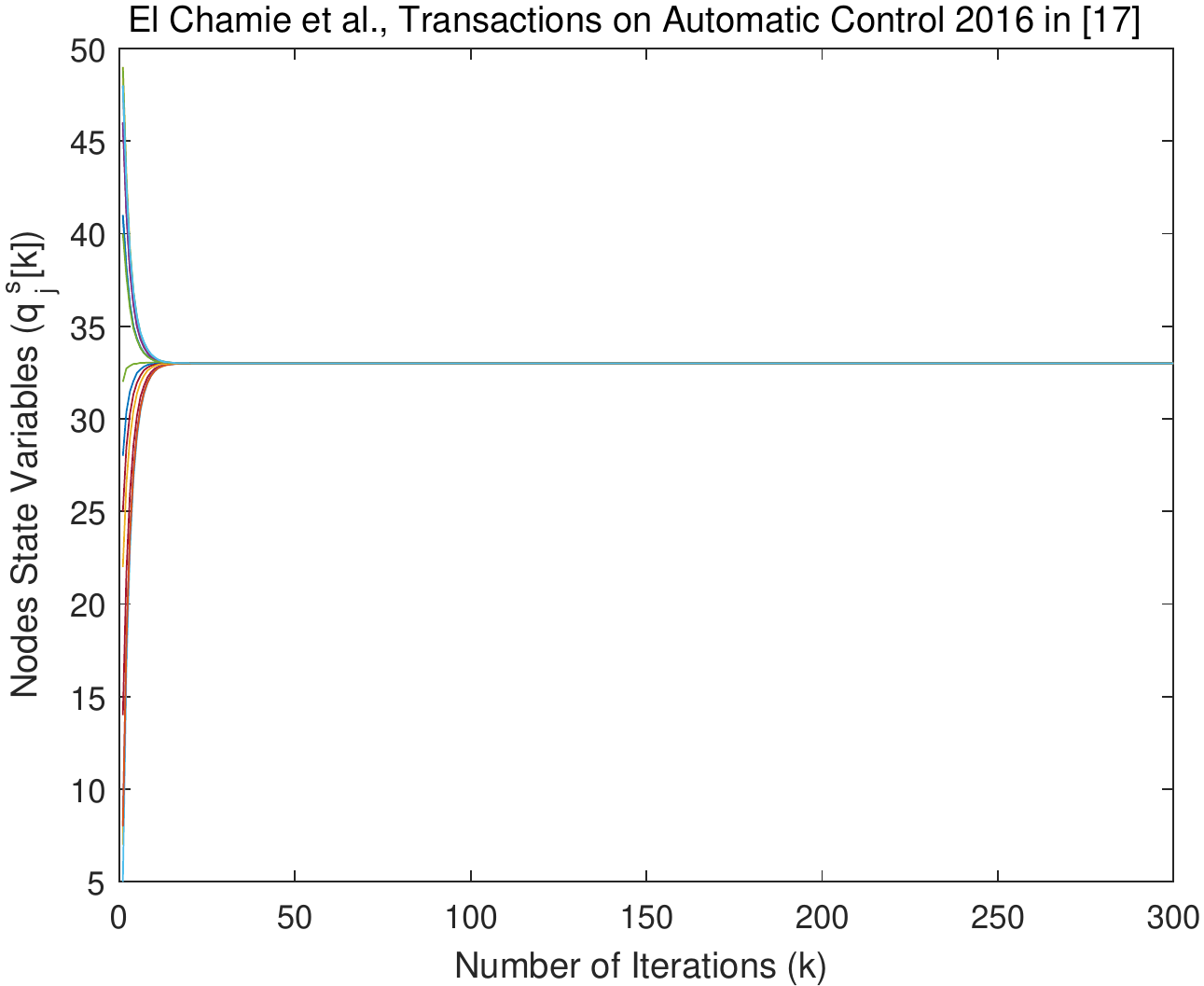}
\caption{Comparison between Algorithm~\ref{algorithm_prob}, the distributed averaging algorithm with quantized communication in \cite{2018:RikosHadj_CDC}, the quantized gossip algorithm presented in \cite{2007:Basar}, the quantized asymmetric averaging algorithm presented in \cite{2011:Cai_Ishii}, and the distributed averaging algorithm with quantized communication presented in \cite{2016:Chamie_Basar} for $1000$ random averaged digraphs of $20$ nodes each.}
\label{comp20}
\end{figure*}

\begin{remark}
It is worth noting, that the quantized distributed algorithms in \cite{2007:Basar, 2011:Cai_Ishii} only involve a single exchange between a single randomly chosen pair of neighboring nodes at each iteration. 
Furthermore, the doubly stochastic matrix which is necessary for the operation of the distributed algorithm in \cite{2016:Chamie_Basar} was formulated by (i) calculating a set of edge weights that balance the given strongly connected digraph with the distributed strategies presented in \cite{2014:RikosHadj} and (ii) by performing a max consensus protocol, adding a nonzero self-loop for every node $v_j$, and normalizing, according to the distributed strategies presented in \cite{2012:CortesJournal}. 
\end{remark}

\section{CONCLUSIONS}\label{future}

We have considered the quantized average consensus problem and presented a randomized distributed averaging algorithm in which the processing, storing and exchange of information between neighboring agents is subject to uniform quantization. 
We analyzed its operation, established that it will reach quantized consensus after a finite number of iterations and argued that its convergence speed appears to be the fastest in the available literature, which allows convergence to the quantized average of the initial values after a finite number of time steps, without any specific requirements regarding the network that describes the underlying communication topology (see \cite{2016:Chamie_Basar}).

In the future we plan to extend the operation of the proposed algorithm to more realistic cases, such as transmission delays over the communication links and the presence of unreliable links over the communication network. 
Furthermore, we plan to design distributed strategies under which every agent in the network will be able to determine whether quantized average consensus has been reached (and thus proceed to execute more complicated control or coordination tasks).

% ------------------------------------------------------------------------------
% Bibliography
% ------------------------------------------------------------------------------
\bibliographystyle{IEEEtran}
\bibliography{bibliografia_consensus}

\balance

\end{document}